\documentclass[12pt]{iopart}
\usepackage{cmap}
\usepackage[utf8x]{inputenc}
\usepackage[T2A]{fontenc}
\usepackage[english,russian]{babel}
\usepackage{graphicx}
\usepackage{hyperref}
\usepackage{iopams}
\usepackage[noEnd=true,indLines=true]{algpseudocodex}

\newtheorem{definition}{Definition}
\newtheorem{remark}{Remark}
\newtheorem{proposition}{Proposition}
\newtheorem{lemma}{Lemma}
\newtheorem{proof}{Proof}

\begin{document}
\selectlanguage{english}
\title[Exact percolation probabilities for a square lattice]{Exact percolation probabilities for a square lattice: Site percolation on a plane, cylinder, and torus}

\author{R~K~Akhunzhanov, A~V~Eserkepov, Y~Y~Tarasevich}
\ead{tarasevich@asu.edu.ru}
\address{Laboratory of Mathematical Modeling, Astrakhan State University, Astrakhan 414056, Russia}

\vspace{10pt}
\begin{indented}
\item[]December 2021
\end{indented}

\begin{abstract}
We have found analytical expressions (polynomials) of the percolation probability for site percolation on a square lattice of size $L \times L$ sites when considering a plane (the crossing probability in a given direction), a cylinder (spanning probability), and a torus (wrapping probability along one direction). Since some polynomials are extremely cumbersome, they are presented as separate files in Supplemental material. The system sizes for which this was feasible varied  up to $L=17$ for a plane, up to $L=16$ for a cylinder, and up to $L=12$ for a torus. To obtain a percolation probability polynomial, all possible combinations of occupied and empty sites have to be taken into account. However, using dynamic programming along with some ideas related to the topology, we offer an algorithm which allows a significant reduction in the number of configurations requiring consideration. A rigorous formal description of the algorithm is presented. Divisibility properties of the polynomials have been rigorously proved. Reliability of the polynomials obtained have been confirmed by the divisibility tests. The wrapping probability polynomials on a torus provide a better estimate of the percolation threshold than that from the spanning probability polynomials. Surprisingly, even a naive finite size scaling analysis allows an estimate  to be obtained of the percolation threshold $p_c = 0.59269$.
\end{abstract}

%
\vspace{2pc}
%
\submitto{\jpa}
%
%
%

\section{Introduction\label{sec:intro}}
Site or bond percolation on graphs including Archimedean lattices~\cite{Suding1999,Scullard2020PRR} and networks~\cite{Cohen2009,Li2021} are some of the most extensively studied problems in statistical physics. Since bond percolation can be treated as site percolation on an edge-to-vertex dual graph, hereinafter we will refer only to the site percolation. Each site of the graph is either occupied with probability $p$ or is empty with a probability $q = 1- p$.

Let $R_L(p)$ be the probability that a finite system of size $L$ percolates at an occupation probability $p$, i.e., the probability that the percolation cluster occurs. When $L \to \infty$,  $R_L(p)$ tends to a step function
\begin{equation}\label{eq:PT}
  R(p) =
\left\{
  \begin{array}{lll}
    0, & \mathrm{if} & p < p_c, \\
    R^\ast, & \mathrm{if} & p = p_c, \\
    1, & \mathrm{if} & p > p_c, \\
  \end{array}
  \right.
\end{equation}
where $p_c$ is the percolation threshold. The value of $R^\ast$ depends on the shape of the system under consideration, on the boundary conditions, and on the way how the percolation is defined. For example, on a torus [a square region with periodic boundary conditions (PBCs)], the occurrence of percolation may be defined when (i)~there exists a cluster that wraps the system along either the horizontal or vertical directions, or both [the percolation probability is denoted as $R_L^{(e)}(p)$]; (ii)~there exists a cluster that wraps the system around one specified direction but not the perpendicular direction [the corresponding percolation probability is denoted as $R_L^{(1)}(p)$]; (iii)~there exists a cluster that wraps the system around both the horizontal and vertical directions [$R_L^{(b)}(p)$]; (iv)~there exists a cluster that wraps the system around the horizontal (vertical) directions [$R_L^{(h)}(p)$ and $R_L^{(v)}(p)$, respectively]~\cite{Newman2000}.

Using conformal field theory, Cardy found the crossing probabilities between segments of the boundary of a compact two-dimensional region at the percolation threshold~\cite{Cardy1992}.
According to Cardy, at the percolation threshold $p = p_c$, the crossing probability $R$ is a function of the rectangle’s aspect ratio $A$
\begin{equation}\label{eq:Cardy}
  R(p_c,A) = \frac{3\Gamma(2/3)}{\Gamma(1/3)^2} m^{1/3} {}_2F_1\left(\left.\frac{1}{3},\frac{2}{3};\frac{4}{3}\right|m \right)
\end{equation}
where ${}_2F_1$ is the hypergeometric function, $m$ is connected with the aspect ratio as
\begin{equation}\label{eq:CardyA}
A = \frac{K(m)}{K(1-m)},
\end{equation}
and $K(m)$ is the elliptic function of the first kind. In particular, when $A=1$ (a square region), obviously, $m = 1/2$ and $R(p_c,1) = 1/2$.

This result has been extended and confirmed by means of computer simulations~\cite{Watts1996,Simmons2007}.  Using conformal field theory, crossing probabilities have also been obtained both for torus~\cite{Pinson1994} and polygonal shapes~\cite{Flores2017}. The compressed result and the values of these probabilities on a torus are presented  in Refs.~\cite{Newman2000,Newman2001} (see also~\cite{Mertens2012} for more accurate values of these probabilities).

Only for few lattices are the exact values of the percolation thresholds known~\cite{Wierman2021}. It is notable that, even for the square lattice, which is the simplest and most extensively studied sample, the exact value of the percolation threshold is known only for bond percolation but not for site percolation.  Several approaches are used to estimate the percolation thresholds in cases when their exact values are not known. First of all, the critical polynomials method should be mentioned~\cite{Scullard2012,Scullard2012b,Scullard2012a,Jacobsen2014,Scullard2020PRR,Xu2021}. In this way, the percolation thresholds have been determined for various lattices~\cite{Scullard2012,Scullard2012a,Scullard2012b,Yang2013,Yang2014,Jacobsen2015,Scullard2020PRR,Scullard2021,Xu2021}. The method is based on consideration of finite size systems. Exact results obtained for these systems can be extrapolated to the thermodynamic limit using finite-size scaling (FSS) analysis.
The graph polynomial method gives $p_c = 0.592\,746\, 01(2)$~\cite{Jacobsen2014} for the site percolation threshold on the square lattice. To simplify the computations, some ideas from game theory~\cite{Yang2013} and Temperley---Lieb algebra~\cite{Jacobsen2015} have been used.
Thus, currently, the most accurate value of the percolation threshold for site percolation on a square lattice $p_c = 0.592\,746\,050\,792\,10(2)$ has been obtained using this method~\cite{Jacobsen2015}. Since its derivation in 2015, this result has not yet been improved on, although it was confirmed by different method to be within the errorbar $p_c = 0.592\,746\,050\,792\,0(4)$~\cite{Mertens2021}. The transfer matrix formalism has been successfully used to obtain the percolation thresholds for a number of cases~\cite{Feng2008,Scullard2012b,Scullard2020PRR,Scullard2021}.

Another possible way to determine the percolation threshold is by calculation of the percolation polynomials, i.e., enumeration of all the possible distributions of occupied and empty sites in a  small system along with exact computation of the percolation probability for each value of $p$~\cite{Ziff1992,Mertens2021}. However, a complete enumeration of all configurations is possible only for very small systems, so this significantly limits the applicability of the method. Limitations arise due both to the time complexity of the problem and the need to store huge data sets, leading to the need for using approximation approaches for estimating the coefficients of the polynomial, for example by using Monte Carlo methods~\cite{Ziff2021}.

Since the number of configurations to be analysed grows as $2^{L^2}$, only small systems can be considered due to the exponential growth of the computations with any increase in the system size~\cite{Yang2013,Yang2014,Mertens2021}. Nevertheless, this method is extremely attractive due to the probably illusory hope of guessing a  regularity that would provide the possibility for constructing a recurrent formula to find the probabilities without directly enumerating all  of the configurations. Unfortunately, as presented in Ref.~\cite{Mertens2021}, detailed analysis of the required time and memory suggests that $L = 24$ is a difficult to overcome limit in the case of a square lattice on a plane.

Moreover, Monte Carlo simulation is extensively used to estimate percolation thresholds. In this method, an estimate of the percolation threshold, $p_c(L)$,  is obtained statistically for several system sizes, $L$. Then, the sequence of $p_c(L)$  is extrapolated to the thermodynamic limit. Notice, that this method has successfully been applied to both discrete~\cite{Newman2000,Newman2001} and continuous~\cite{Mertens2016} percolation problems. However, this method does also have some limitations. The precision of the estimate of the percolation threshold for any given system size depends on the number of independent runs. This statistical error cannot be completely eliminated. Reducing the statistical error is possible only by increasing the number of runs, but this number cannot be increased infinitely due to the  limitations of the reasonable simulation times. Indeed, increase in the system size leads to a significant increase in the required simulation time, while its effect on the extrapolation accuracy is only  modest. Although, over a number of decades, the Monte Carlo approach has played a crucial role in estimations of percolation thresholds, it appears that, currently, straight Monte Carlo has completely exhausted its possibilities~\cite{Ziff2021}. Even the combination of exact enumerations with Monte Carlo simulations~\cite{Yang2013} is hardly a lifebuoy.  Thus, combination of exact percolation  polynomials for small systems along with numerical estimations for larger systems gives $p_c = 0.592\,746\,050\,95(15)$, i.e.,  9 accurate digits~\cite{Yang2013}. Extensive review of the results and methods for the determination of $p_c$ for site percolation on a square lattice (including some not mentioned above) can be found in Ref.~\cite{Ziff2011}.

There are several different ways to estimate the percolation threshold. An excellent review of these methods can be found in Ref.~\cite{Ziff2002}. Only three of  these estimates will be used in the present study.
\begin{enumerate}
  \item The estimate $p^\ast$ corresponding to the point where $R_L(p)$ equals its universal value in the thermodynamic limit $R^\ast$~\cite{Ziff1992}
\begin{equation}\label{eq:estimator-ast}
R_L(p^\ast) = R^\ast.
\end{equation}
  \item The estimate $p^{infl}$ corresponding to the point where $R'_L(p)$ reaches a maximum (or equivalently, where $R_L(p)$ is at its inflection point)~\cite{Reynolds1980}
\begin{equation}\label{eq:estimator-infl}
R''_L(p^{infl}) = 0,
\end{equation}
here, the primes indicate differentiation with respect to $p$.
    \item The estimate $p^{cc}$ corresponding to the point where two systems of different size have the same value of $R$~\cite{Reynolds1980}
        \begin{equation}\label{eq:estimator-cc}
        R_{L1}(p^{cc}) = R_{L2}(p^{cc}),
        \end{equation}
        where $L1$ and $L2$ may be, e.g., $L$ and $L-1$ or $L$ and $L/2$.
\end{enumerate}

Since the percolation threshold is related to the thermodynamic limit, while both percolation polynomials and the Monte Carlo estimates can be obtained for finite-size systems, an FSS is needed. A well-known scaling relation says
\begin{equation}\label{eq:FSS}
  p_c(L) - p_c(\infty) \propto L^{-1/\nu},
\end{equation}
where $\nu$ is the critical exponent ($\nu = 4/3$ in 2D)~\cite{Stauffer2018}.

When the percolation threshold is estimated on a torus using $R^\ast$, one can significantly reduce the amount of computation needed, since the convergence to the thermodynamic limit of the percolation threshold values obtained for systems of finite size is high. Hence, this is the most efficient method to estimate the percolation threshold~\cite{Newman2000,Newman2001,Li2009,Mertens2012}, since
\begin{equation}\label{eq:scaling}
p_c(L) - p_c(\infty) \propto L^{-2-1/\nu}.
\end{equation}
Using FSS, the value of the percolation threshold in the thermodynamic limit can be obtained. In any case, this FSS requires the limit $L \to \infty$ which will hardly hold when percolation polynomials are considered. There have been numerous attempts to improve FSS. A widely used ansatz for the asymptotic behaviour of $p_c(L)$, motivated by the general principles of
FSS, is that of a series of power-law corrections
\begin{equation}\label{eq:FSSseries}
  p_c(L) = p_c(\infty) + \sum_{k=1}^{\infty}A_k L^{-\Delta_k},
\end{equation}
where all $\Delta_k > 0$ and $\Delta_k < \Delta_{k+1}$~\cite{Jacobsen2015,Mertens2016,Mertens2017,Mertens2021}.
In physics, the Bulirsch--Stoer algorithm~\cite{Bulirsch1964} is widely used for FSS~\cite{Henkel1988,Monroe2002} including for percolation~\cite{Mertens2021}.

Mostly, open border systems are used to calculate the percolation probabilities. Since the systems under consideration are small, the finite-size effect is significant. The effect of boundaries can be reduced by considering PBCs along one direction (percolation on a cylinder) or along both directions, i.e., percolation on a torus. The percolation polynomials on a torus have been computed by Mertens~\cite{MertensHomePage} up to $L=11$, unfortunately, however, they have not been published as an article. Moreover, for $L>7$, these polynomials fail the divisibility test. The use of PBCs seems to be very promising, since the effectiveness of this approach has repeatedly been proven in works where the percolation thresholds have been estimated by the Monte Carlo method~\cite{Newman2000,Newman2001}. The transfer of the approach to the case of finding the percolation threshold using percolation polynomials looks potentially fruitful.

The goal of the present work is the computation of percolation probabilities on a plane, on a cylinder, and on a torus. The rest of the paper is constructed as follows. \Sref{sec:methods} describes some technical details of the simulations. Proof of the divisibility property of the coefficients of the percolation polynomials is presented in \ref{subsec:DT}. \Sref{sec:results} presents our main findings. \Sref{sec:concl} summarizes the main results. The mathematical background of the algorithm is presented in \ref{sec:BG}. The obtained percolation polynomials are presented in the Supplementary material.

\section{Methods\label{sec:methods}}
\subsection{Common information}
We were looking for analytical expressions (polynomials) of percolation probability for site percolation on a square lattice of size $L \times L$ sites
$$
R_L^{(v)}(p) = \sum_{i=0}^{L^2} c_i p^i q^{L^2 - i}
$$
considering  a square region with different boundary conditions, viz., (i)~PBCs along both directions, i.e., a torus (wrapping probability along one direction), (ii)~PBC along one direction, i.e., a cylinder (spanning probability), and (iii)~open boundaries, i.e., a plane (spanning probability). The latter case was used as a test of our algorithm and software, since the percolation polynomials for this case have been previously published~\cite{Ziff2002,MertensHomePage}. The system size varied up to $L=12$ in the case of the torus, up to $L=16$ in the case of the cylinder, and up to $L=17$ in the case of the plane. According to the method presented in Ref.\cite{Mertens2021}, the low estimate of the number of configurations to be considered on torus $L \times L$ is a square of the number of those for a plane of the same size. The number of configurations for the torus $L=12$ approximately corresponds to those for the plane $L=21$, i.e., exceeds $10^9$.

To obtain a percolation probability polynomial, all possible combinations of occupied and empty sites were taken into account. We used an algorithm  based on dynamic programming along with some ideas from topology,  which allowed us significantly to reduce the number of configurations under consideration. Although the algorithm is close to those of other authors~\cite{Yang2013,Yang2014,Mertens2021}, an independent implementation of the algorithm was used. The mathematical background of the algorithm is presented in~\ref{sec:BG}. In fact, a formal description and justification of a family of algorithms is presented. This family of algorithms is applicable to a wide range of problems, where there are a finite number of objects each of which may be in the two states with probabilities $p$ and $q=1-p$, respectively. The description is based on the probability theory. The algorithms described in terms of transfer matrix belong to the same family of algorithms. However, our formal description needs no mention of transfer matrix, game theory or Temperley---Lieb algebra. This gives a new perspective on the problem and some freedom to apply the algorithm.

We used and compared all three estimators~\eref{eq:estimator-ast}, \eref{eq:estimator-infl}, and \eref{eq:estimator-cc}.

We  used the C++ bignum library~\cite{bignum}. All computations were implemented on a PC (Intel\textsuperscript{\circledR} Xeon\textsuperscript{\circledR} E5-2690 v3 CPU, with a CPU clock speed 2.6 GHz and 256 GB RAM).

For a plane $L \leqslant 11$, the spanning probability polynomials coincide with the previously published results~\cite{Ziff2002,MertensHomePage,Mertens2021}. As an additional test, the divisibilities of the coefficients of the obtained polynomial were used.

\subsection{Divisibility test}\label{subsec:DT}

Consider an event $A$ that is the percolation along the vertical direction in a rectangle $L_{1} \times L_{2}$. PBCs may be applied along one or both directions. Let $k:0\leqslant k\leqslant N$. Let $\Omega _{k}$ be an event where the status of the percolation state $S$ is positive and $\#_{1}( S)=k$. Let $G$ be a group of all the translations on this rectangle taken into account the PBCs. Consider an action of the group $G$ on the set $\Omega _{k}$. Due to the action of the group $G$, the set $\Omega _{k}$ splits over several non-intersecting orbits, $O_{S_{i}}$, of the situation $S_{i} ( 1\leqslant i\leqslant m)$: $\Omega _{k} =\bigsqcup _{i=1}^{m} O_{S_{i}}$, here $m$ is the number of orbits. Thus, $|\Omega _{k} |=\sum _{i=1}^{m} |O_{S_{i}} |$.

$G_{S_{i}}$ is a stabilizer of the situation $S_{i}$, i.e., a subgroup of the group $G$. 
According to the orbit-stabilizer theorem $|O_{S_{i}} |=\frac{|G|}{|G_{S_{i}} |}$.
In our case, $\#_{1}( S_{i}) = k$.

\begin{proposition}\label{prop:number}
$
\#_{1}( S_{i}) \vdots |G_{S_{i}} |.
$
\end{proposition}
Here, $N \vdots M$ means that the integer $N$ is divisible by the integer $M$.

\begin{proposition}\label{prop:Lagrange}
According to Lagrange's theorem, $|G|\vdots |G_{S_{i}} |$.
\end{proposition}

From propositions~\ref{prop:number} and \ref{prop:Lagrange}, it follows  that
$\gcd( \#_{1}( S_{i}) ,|G|) \vdots |G_{S_{i}} |.$ Hence, there is an integer number $n_{S_{i}}$ such that
$\gcd(\#_{1}( S_{i}) ,|G|) =n_{S_{i}} |G_{S_{i}} |,
$
$$
|G_{S_{i}} |=\frac{\gcd( \#_{1}( S_{i}) ,|G|)}{n_{S_{i}}}=\frac{\gcd( k,|G|)}{n_{S_{i}}}.
$$
Obviously $|G|\vdots \gcd( k,|G|)$. Hence, there is an integer number  $\ell _{k}$ such that
$|G|=\ell _{k} \gcd( k,|G|),$
where
$$
\ell _{k} =\frac{|G|}{\gcd( k,|G|)}.
$$

\begin{lemma}
$$
c_{k} \vdots \frac{|G|}{\gcd( k,|G|)}.
$$
\end{lemma}

\begin{proof}
 $$
 c_{k} =|\Omega _{k} |=\sum _{i=1}^{m} |O_{S_{i}} |=\sum _{i=1}^{m}\frac{|G|}{|G_{S_{i}} |} =
 \sum _{i=1}^{m} \ell _{k} n_{S_{i}} =\ell _{k} \sum _{i=1}^{m} n_{S_{i}} \vdots \ell_{k}. \qquad\Box
$$

\end{proof}
\begin{remark}
For a torus (PBCs along both directions),  $|G|=L_{1} L_{2}$, while for a cylinder (PBC along one direction, say 1),   $|G|=L_{1}$.
\end{remark}
\begin{remark}
  The consideration above is valid for any event $A$ that is invariant with respect to a translation group.
\end{remark}

Thus, when percolation on a torus is considered, the quantity
$$
\frac{L^2}{\gcd(i,L^2)}
$$
must be a divisor of $c_i$. Here, $\gcd$ means  the greatest common divisor. All obtained polynomials passed this test. Likewise, in the case of the cylinder,
$$
\frac{L}{\gcd(i,L)}
$$
must be a divisor of $c_i$.

\section{Results\label{sec:results}}
\subsection{Percolation on a cylinder\label{subsec:cylinder}}

To reduce the boundary effect, we applied PBC along one direction, i.e., considered the percolation on a cylinder. The percolation polynomials have previously been found  for spanning. In line with our expectations, the estimates of the percolation thresholds obtained using the spanning percolation probability on a torus are better than those obtained for a plane.
\Tref{tab:cylinder} presents estimates $p^{infl}$  and $p^{cc}$ for a cylinder.
\begin{table}
\centering
\caption{Estimates $p^{infl}$ and $p^{cc}$ for a cylinder.}\label{tab:cylinder}
\begin{tabular}{lll}
\br
$L$ & $p^{infl}$ &  $p^{cc}$ \\
\mr
 3 & 0.567797933184318829071364943136 & \\
 4 & 0.564723542379152649640582426903 & 0.575695178318265035538397717037 \\
 5 & 0.564394503892627964586877823099 & 0.577014627673463823472571092197 \\
 6 & 0.565139995218697209864023095252 & 0.579807628772524386223709742598 \\
 7 & 0.566243473648386250544306475856 & 0.582009179107144070215814399129 \\
 8 & 0.567430070283077378434607247195 & 0.583739124639758575223057384096 \\
 9 & 0.568591188366745346953117076744 & 0.585091109091358033925618933501 \\
10 & 0.569685554282538273261796521766 & 0.586161181920254733599409764592 \\
11 & 0.570700318144334727392544417721 & 0.587019710493454014651253291769 \\
12 & 0.571634812621509476323690300113 & 0.587717903064583897066422803208 \\
13 & 0.572493438058381621488860472585 & 0.588292943898443877174066922998 \\
14 & 0.573282448954560869061470552053 & 0.588772093765971493025856683878 \\
15 & 0.574008491924559902782966679906 & 0.589175584440447625295942237502 \\
16 & 0.574677957108836363064812168494 & 0.589518631123110527410156176131 \\
\br
\end{tabular}
\end{table}

\Fref{fig:estimates-cyl} and \tref{tab:cylinder} demonstrates behaviours of estimates $p^{infl}$ and $p^{cc}$ for a cylinder.
\begin{figure}[!htb]
  \centering
  \includegraphics[width=0.6\columnwidth]{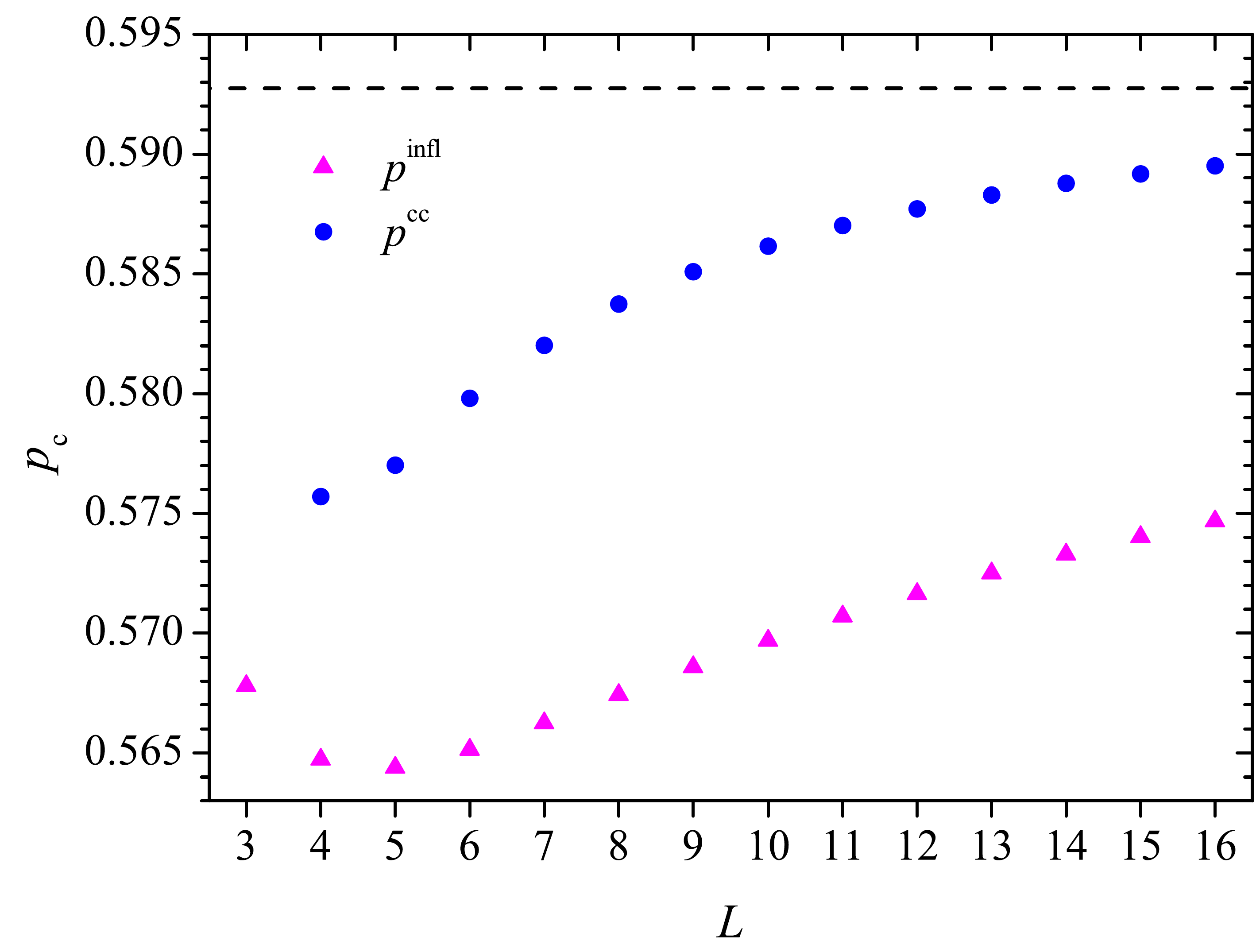}
  \caption{Estimates of the percolation threshold for cylinder $p^{infl}$ (\fulltriangle), and $p^{cc}$ (\fullcircle) plotted against the system size, $L$.   Dashed line corresponds to the most accurate known value of the percolation threshold $p_c = 0.592\,746\,050\,792\,10(2)$~\cite{Jacobsen2015}. \label{fig:estimates-cyl}}
\end{figure}

\Fref{fig:RvspCylinder} demonstrates the percolation (spanning) probabilities $R_L(p)$ for a cylinder, $L \in [3,16]$.
\begin{figure}[!htb]
  \centering
  \includegraphics[width=0.5\columnwidth]{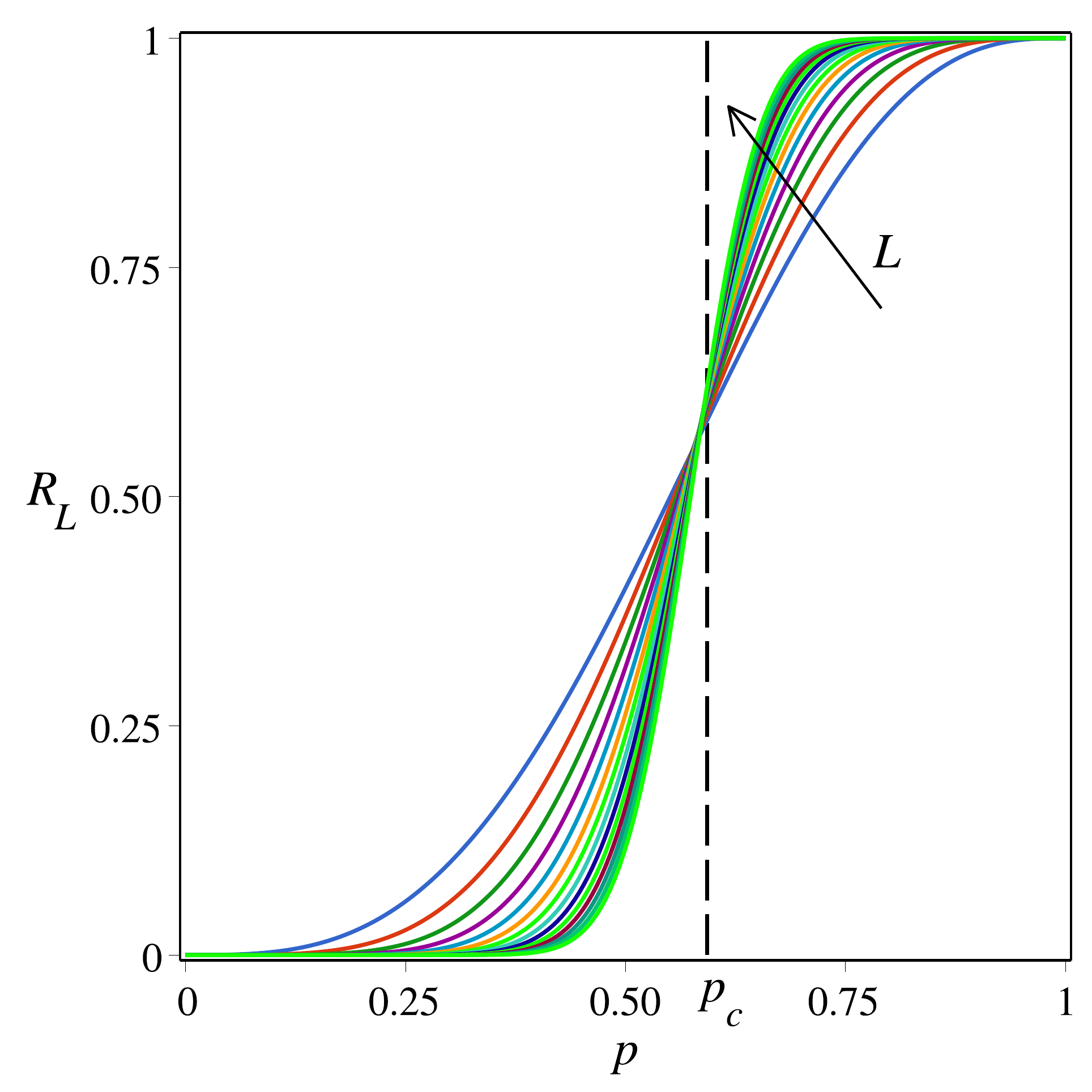}\includegraphics[width=0.5\columnwidth]{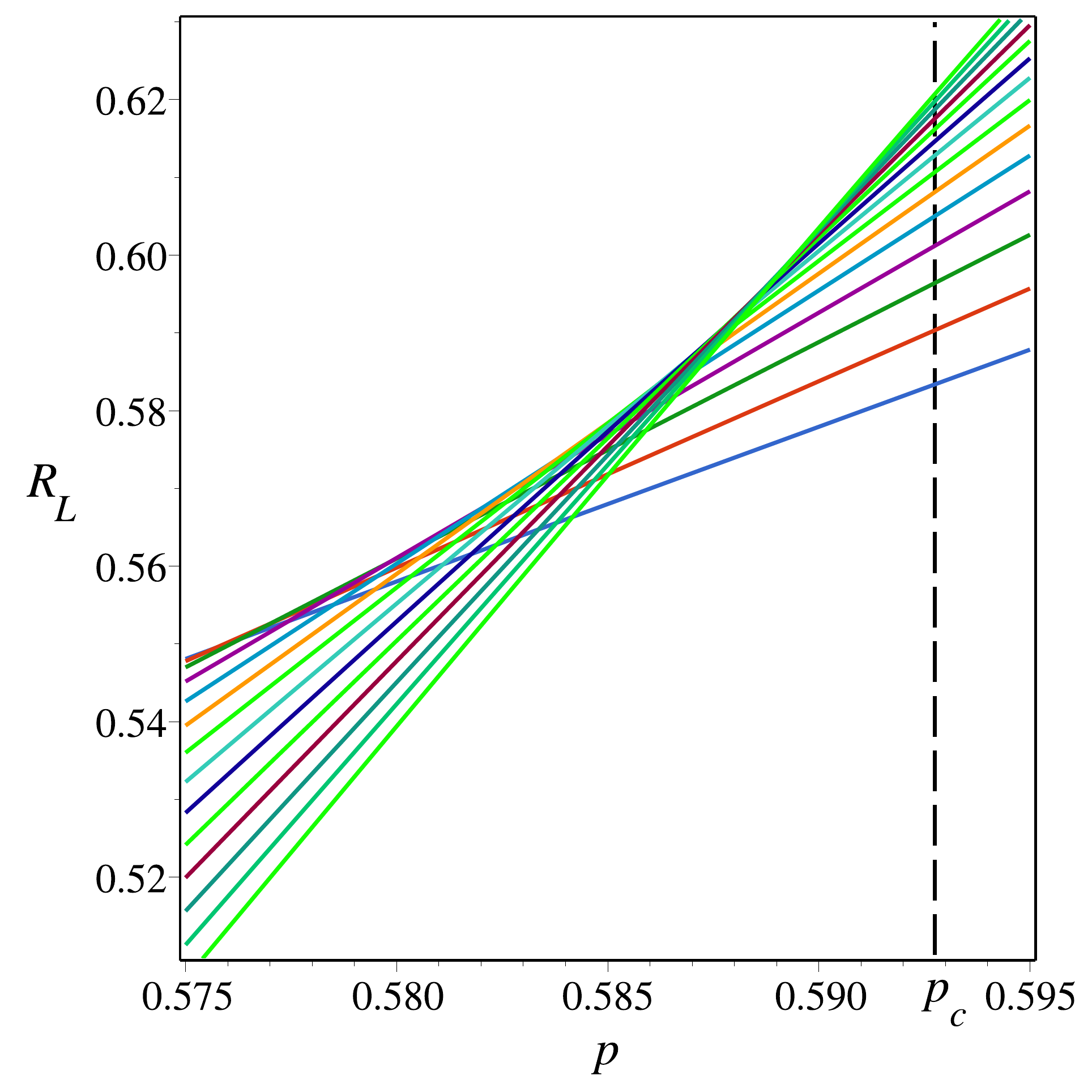}
  \caption{Percolation probabilities $R_L(p)$ for a cylinder; full view (left panel) and enlarged region near the percolation threshold (right panel). The larger the system size, the sharper the step.\label{fig:RvspCylinder}}
\end{figure}

\subsection{Percolation on a torus\label{subsec:torus}}

\Fref{fig:RvspTorus} demonstrates the percolation probabilities $R_L(p)$ for a torus, $L \in [3,12]$. Our polynomials for $L \in [3,7]$ match exactly the results by Mertens~\cite{Ziff2002,MertensHomePage}. It seems that $L=12$ is the limit that can be reached within a reasonable computation time. This limit is unlikely to be overcome in the near future.
\begin{figure}[!htb]
  \centering
  \includegraphics[height=0.32\textheight]{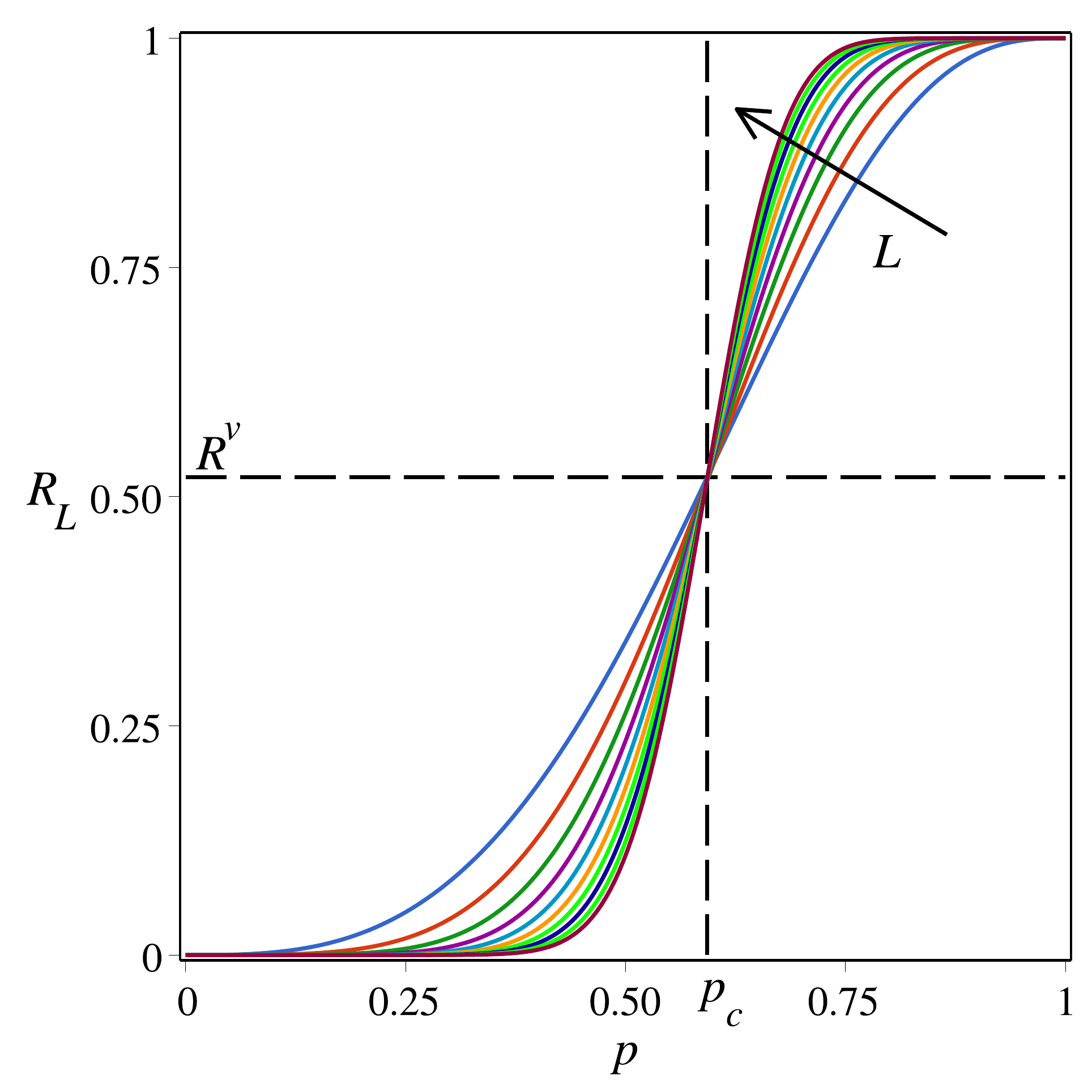}\includegraphics[height=0.32\textheight]{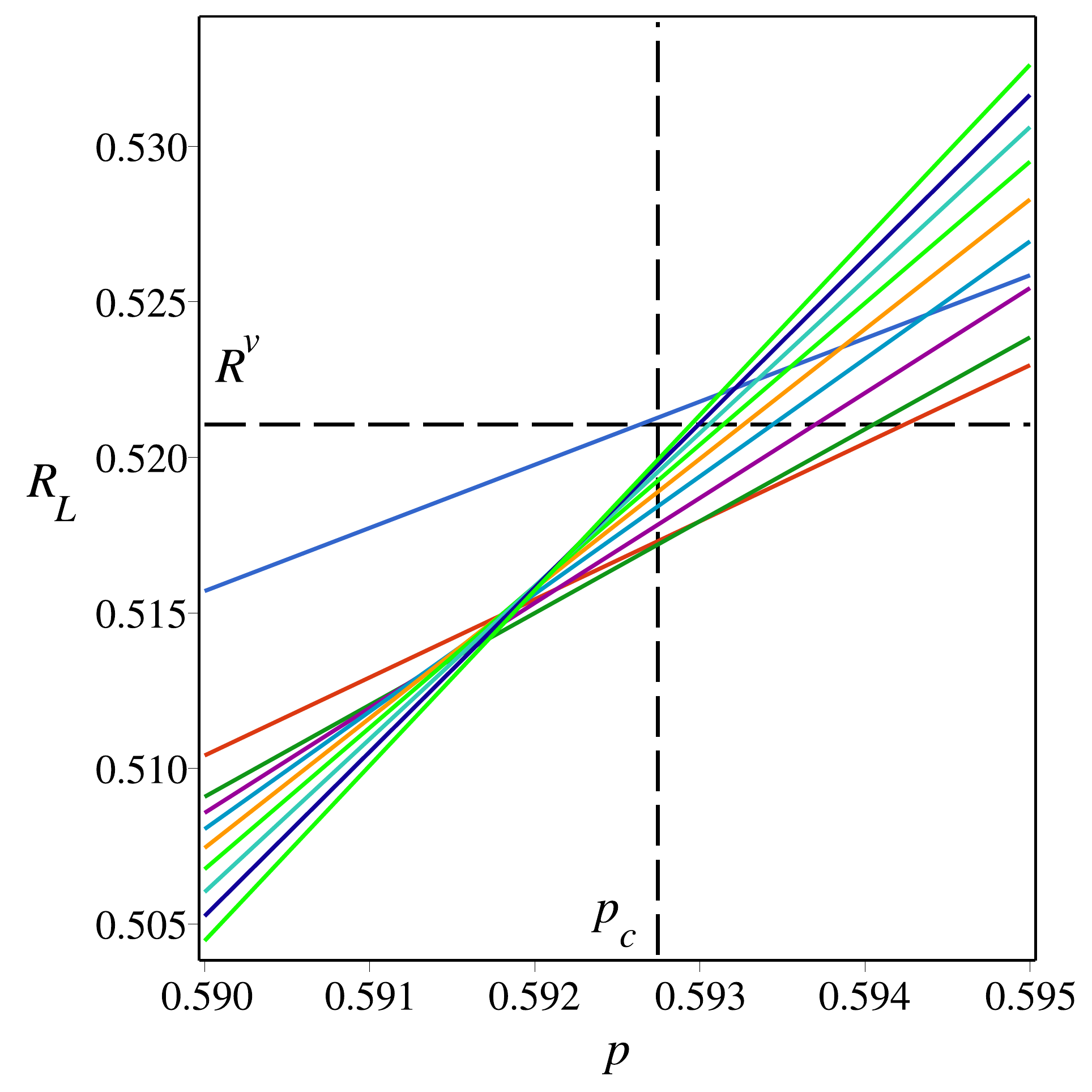}
  \caption{Percolation probabilities $R_L(p)$ for a torus; full view (left panel) and enlarged region near the percolation threshold (right panel). The horizontal dashed line corresponds to the value $R^v = 0.5210\dots$~\cite{Mertens2016}, while the vertical dashed line corresponds to the value $p_c = 0.5927\dots$. The  larger the system size, the sharper the step.\label{fig:RvspTorus}}
\end{figure}

The first derivatives of the percolation probabilities $R_L(p)$ for the torus are shown in~\fref{fig:dRdpvspTorus}.
\begin{figure}[!htb]
  \centering
  \includegraphics[width=0.5\columnwidth]{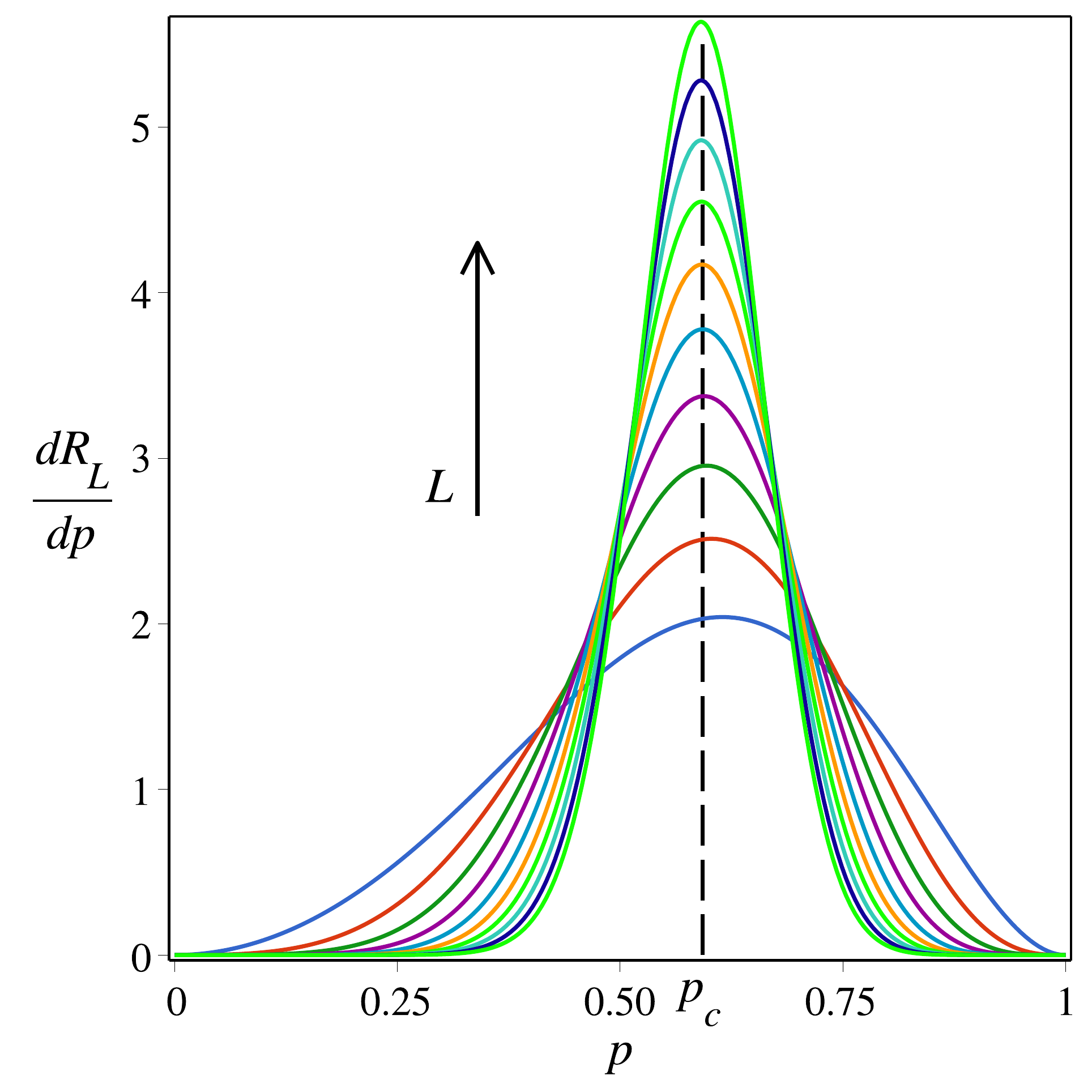}
  \caption{First derivatives of the percolation probabilities $R_L(p)$ for a torus. The larger the system size, the higher the maximum. \label{fig:dRdpvspTorus}}
\end{figure}

\Fref{fig:d2Rdp2vspTorus} shows the second derivatives of the percolation probabilities $R_L(p)$ for the torus.
\begin{figure}[!htb]
  \centering
  \includegraphics[height=0.32\textheight]{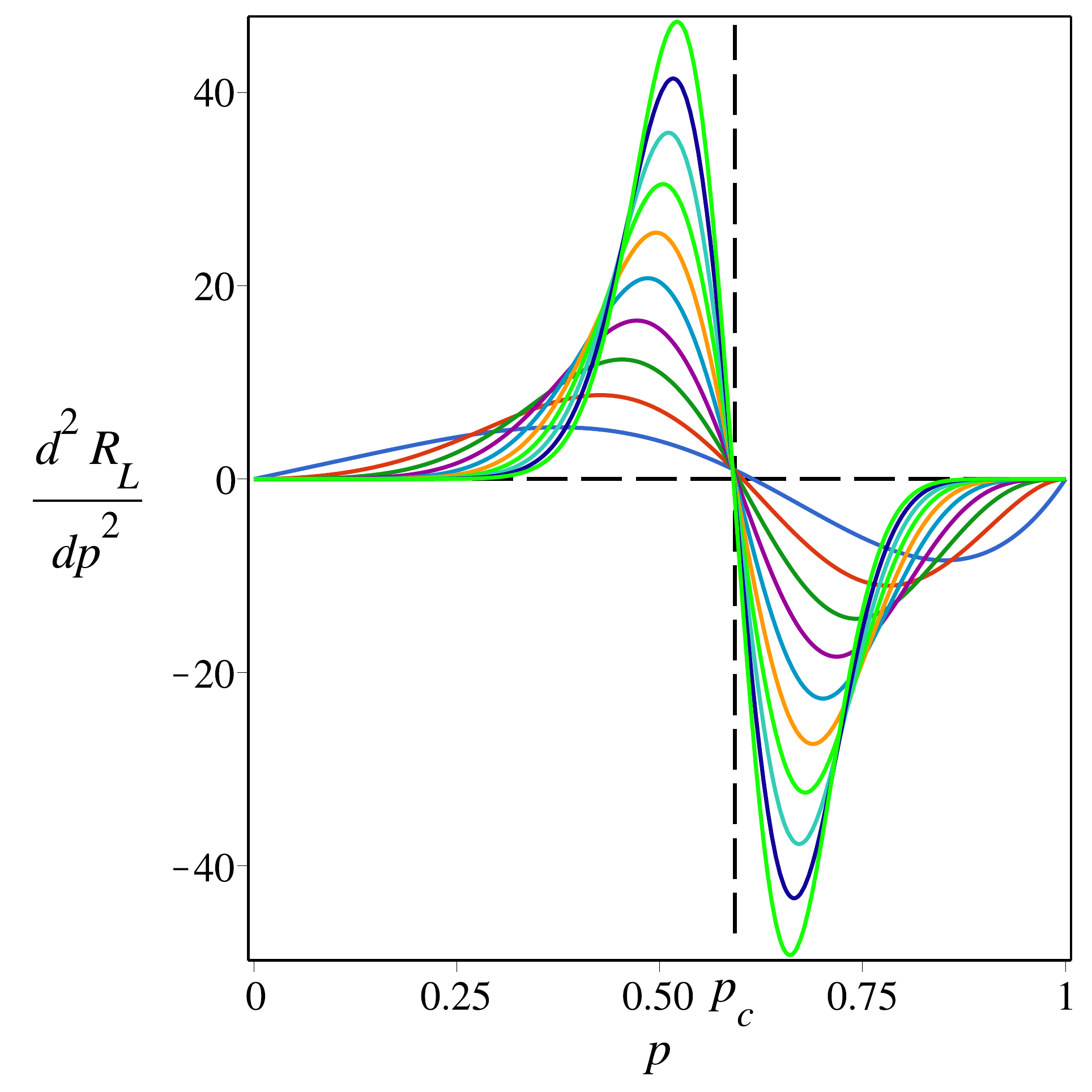}\includegraphics[height=0.32\textheight]{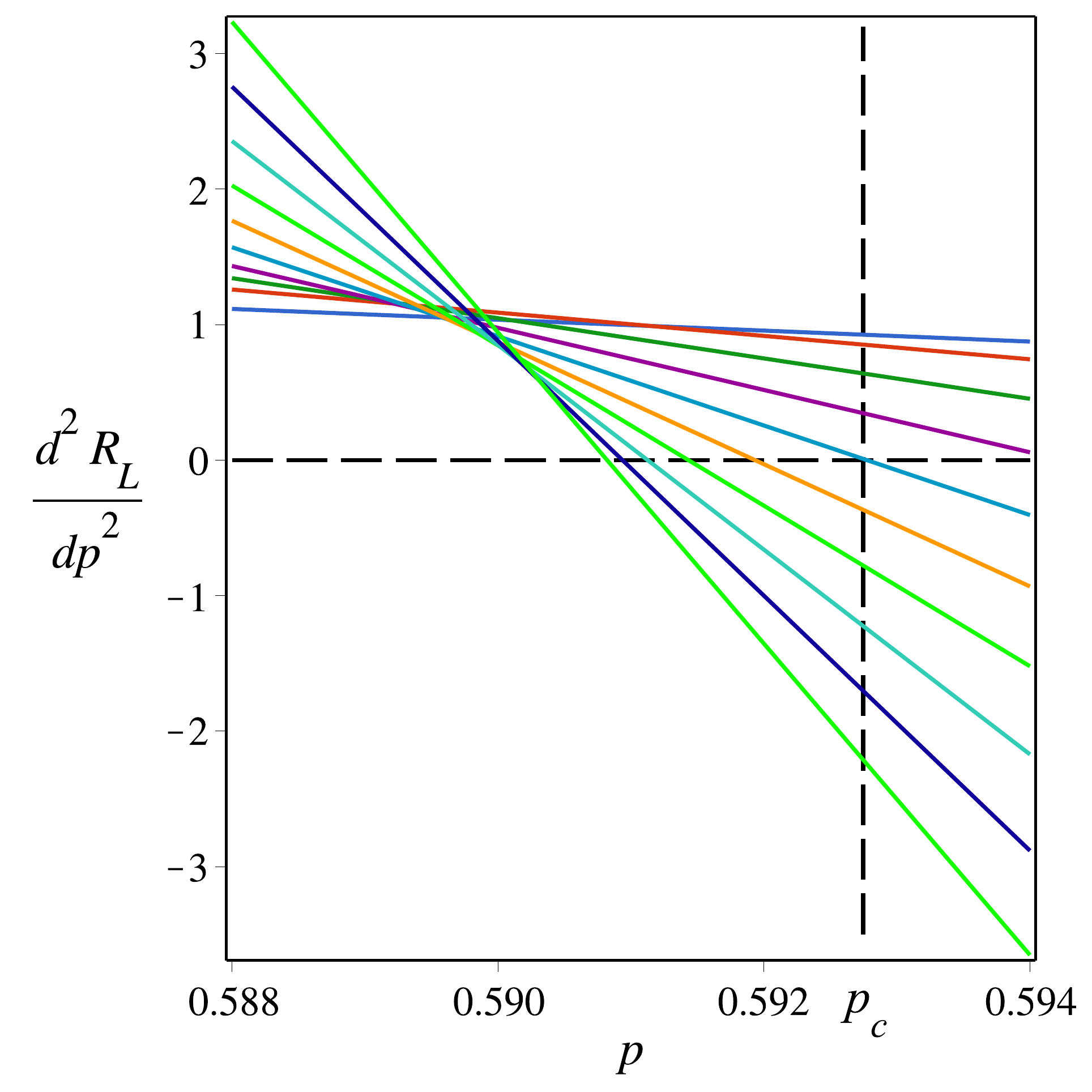}
  \caption{Second derivatives of the percolation probabilities $R_L(p)$ for a torus; full view (left panel) and enlarged region near the percolation threshold (right panel). The larger the system size, the  higher the maximum.\label{fig:d2Rdp2vspTorus}}
\end{figure}

\Tref{tab:torus} presents estimates $p^{\ast}$, $p^{infl}$,  and $p^{cc}$ for a torus.
\begin{table}
\centering
\caption{Estimates $p^{\ast}$, $p^{infl}$, and $p^{cc}$ for a torus.}\label{tab:torus}
\begin{tabular}{lll}
\br
$L$ & $p^{\ast}$ &  $p^{cc}$ \\
\mr
 3 & 0.592639952553406926057811117921 & \\
 4 & 0.594241786767314544427579244396 & 0.601048018206869318922976758793\\
 5 & 0.594053483642978334127033318840 & 0.592989260418213921286737949416\\
 6 & 0.593701218965827543995008913020 & 0.591226641078777884419327732025\\
 7 & 0.593442713470113120566237227663 & 0.591284821573047295994536552773\\
 8 & 0.593265367086649457890784962079 & 0.591551386819268445717847740778\\
 9 & 0.593142712666508735278138783115 & 0.591795816259173927085827777018\\
10 & 0.593055910631289361098943275939 & 0.591987393802541280038522512867\\
11 & 0.592992938143880591354064973534 & 0.592132592840694226990290804247\\
12 & 0.592946179685333445817365777051 & 0.592250990072130062279383528928\\
\br
\end{tabular}
\medskip
\begin{tabular}{ll}
\br
$L$ & $p^{infl}$  \\
\mr
3  & 0.614851397846434431296649483909 \\
4  & 0.602515335071713060222047930819 \\
5  & 0.597021025923632529620024559897 \\
6  & 0.594251029001068374748168858373 \\
7  & 0.592774120165894290989073926981 \\
8  & 0.591933286683125034089609750005 \\
9  & 0.591432548170789434233257646110 \\
10 & 0.591126149630403711426846910202\\
11 & 0.590936706655684238266345408782\\
12 & 0.590821325857797991285938830656\\
\br
\end{tabular}
\end{table}

\Fref{fig:estimates} and \tref{tab:torus} evidence that estimate $p^{\ast}$ is more promising due to its faster convergence.
\begin{figure}[!htb]
  \centering
  \includegraphics[width=0.6\columnwidth]{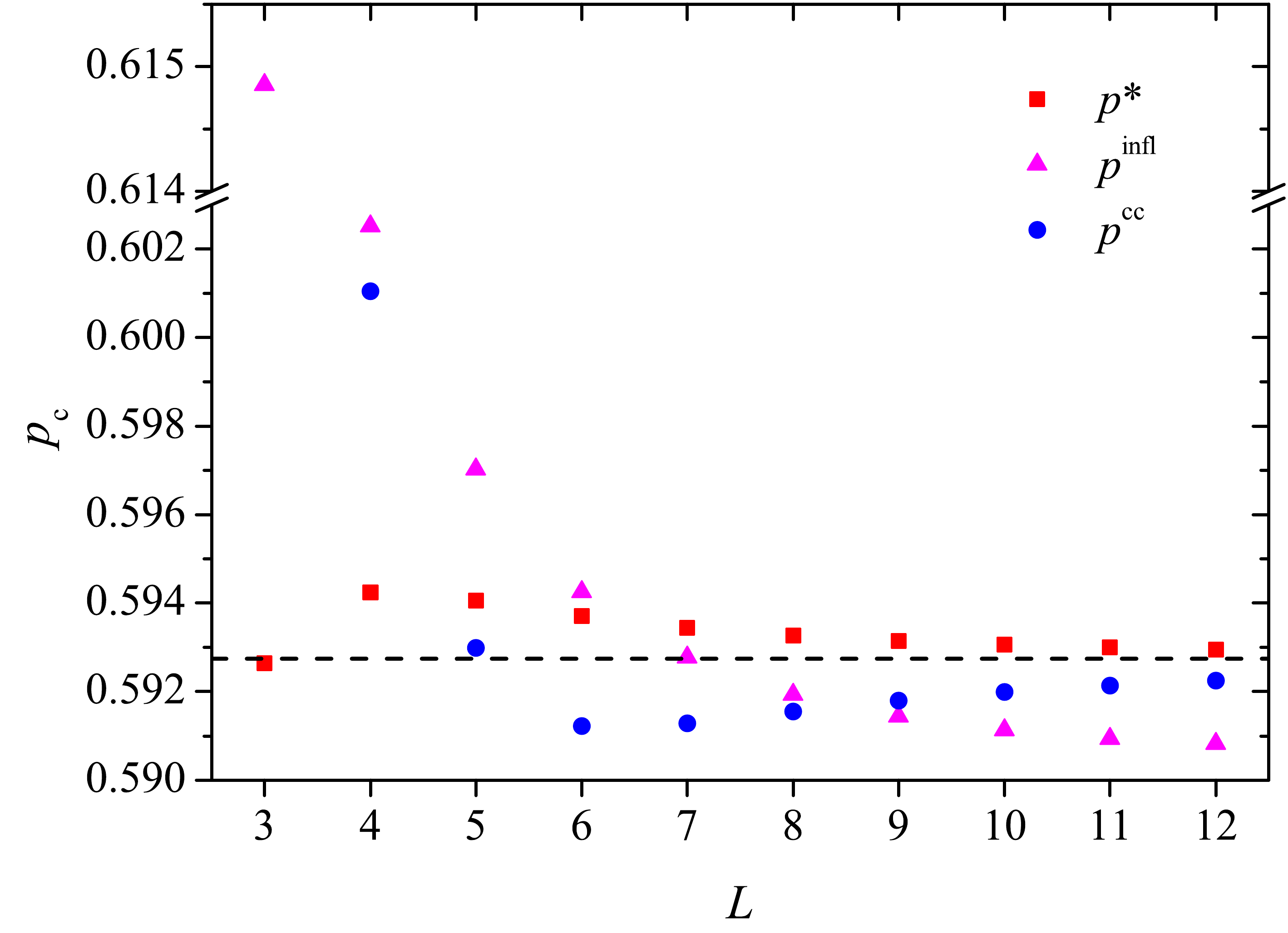}
  \caption{Estimates of the percolation threshold for torus, $p^\ast_c$  (\fullsquare), $p^{infl}$ (\fulltriangle), and $p^{cc}$ (\fullcircle) plotted against the system size, $L$.   Dashed line corresponds to the most accurate known value of the percolation threshold $p_c = 0.592\,746\,050\,792\,10(2)$~\cite{Jacobsen2015}. \label{fig:estimates}}
\end{figure}

Using the estimate $p^\ast_c$, the simplest FSS~\eref{eq:FSSseries} with only two first terms leads to $p_c \approx 0.59273(3)$ (\fref{fig:FSS}), since with the available number of points suitable for extrapolation, the following terms in~\eref{eq:FSSseries}  do not affect the value of the percolation threshold.
\begin{figure}[!htb]
  \centering
  \includegraphics[width=0.6\columnwidth]{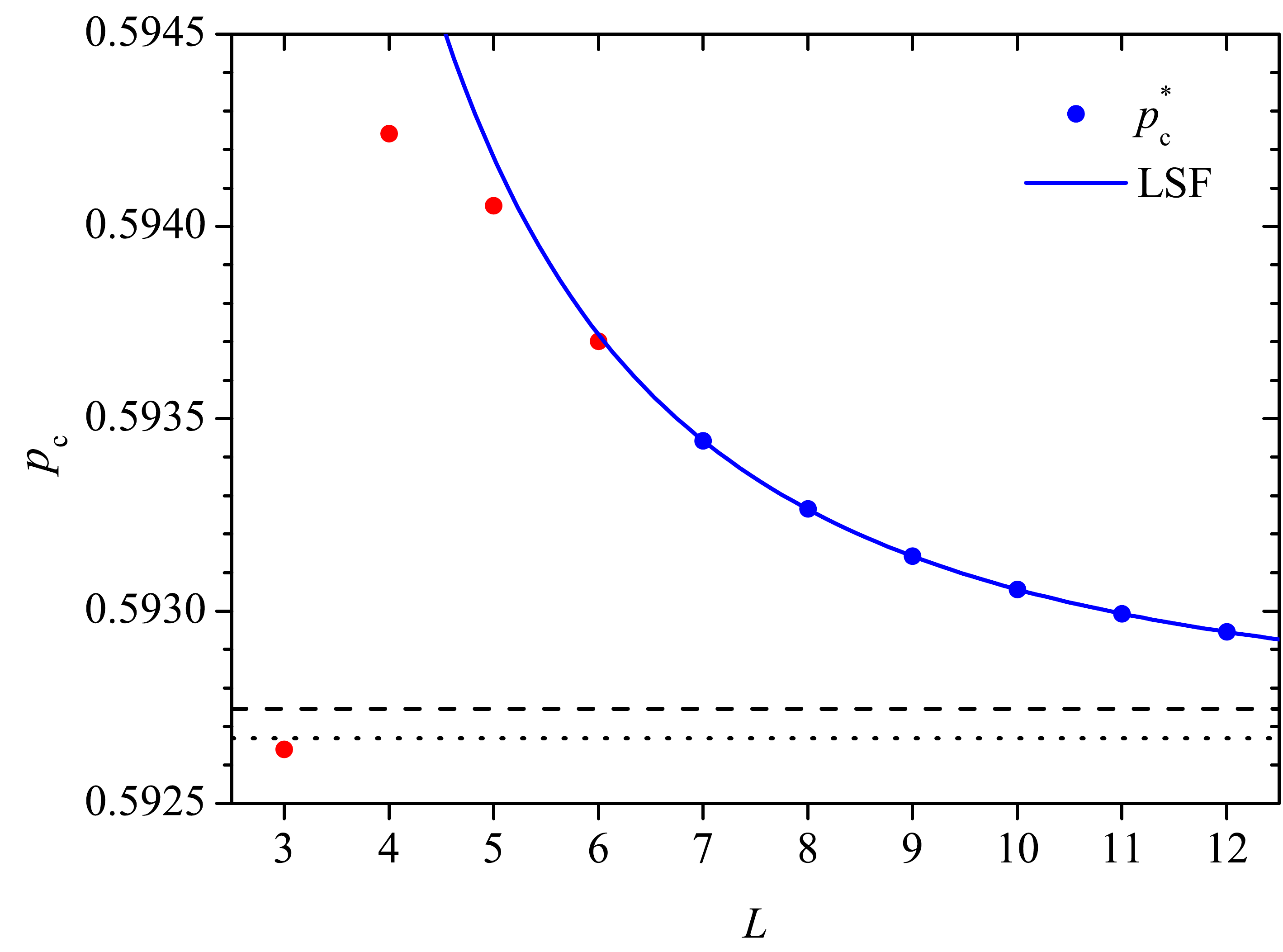}
  \caption{FSS for $p^\ast_c$ on a torus. Dashed line corresponds to the most accurate known value of the percolation threshold~\cite{Jacobsen2015}.   Dotted line corresponds to our value of the percolation threshold $p_c = 0.59269$.\label{fig:FSS}}
\end{figure}

\subsection{Comparison}\label{subsec:comparison}

\Fref{fig:TorusCylinderPlane} compares the dependencies of the percolation threshold, $p_c$, for a plane, a cylinder, and a torus. Even without any FSS, estimates of the percolation threshold obtained for the torus converge to the percolation threshold in the thermodynamic limit much faster when compared to those obtained for the plane and for the cylinder.
\begin{figure}[!htb]
  \centering
  \includegraphics[width=0.6\columnwidth]{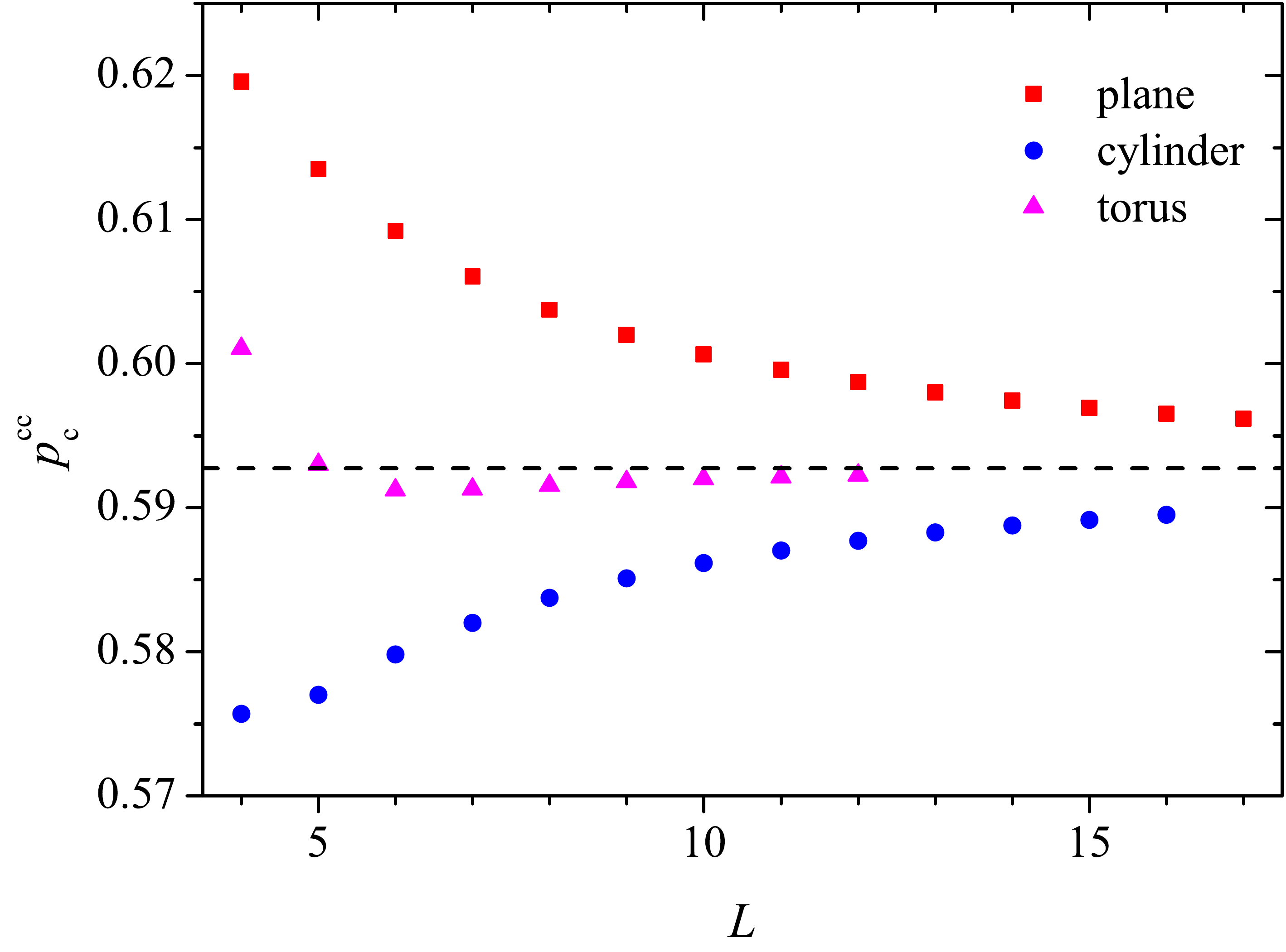}
  \caption{Estimates of the percolation threshold, $p_c^{cc}$, for a plane (\fullsquare), a cylinder (\fullcircle), and a torus (\fulltriangle) plotted against the system size, $L$.   Dashed line corresponds to the most accurate known value of the percolation threshold~\cite{Jacobsen2015}.\label{fig:TorusCylinderPlane}}
\end{figure}

\section{Conclusion\label{sec:concl}}

We studied site percolation on a square lattice. We found the percolation polynomials for a square region with (i)~open boundaries (a plane), (ii)~PBC along one direction (a cylinder), and (iii)~PBCs along both mutually perpendicular directions (a torus). For the plane, the percolation polynomials obtained are exactly the same as those obtained by other authors~\cite{Ziff2002,MertensHomePage}, which is an indirect confirmation of the correctness of our results for the two other cases. Further, we proposed a divisibility test. All the obtained polynomials passed this test, providing further confirmation of the correctness of our results.

Three different estimates of the percolation threshold were used. We found that the estimate
corresponding to the point, where the percolation polynomial equals its universal value in the thermodynamic limit, provides a faster convergence than the other estimates.  We found that, as the system size increased, any estimate of the percolation threshold exhibited faster convergence when PBCs were applied to the system.

It seems that both the Monte Carlo method and the method of percolation polynomials have achieved the utmost accuracy in determining the percolation threshold. Further refinement of the percolation threshold values is hardly possible by using these methods. Although the use of percolation polynomials detected some regularities in the coefficients of the polynomials~\cite{Mertens2019,Mertens2021}, these regularities have not yet enabled to propose a formula for finding all the coefficients without laborious calculations. Even very sophisticated and difficult FSS methods also seem to have reached their limits.

\ack
We dedicate this paper to Prof. Robert M.~Ziff, whose numerous influential papers devoted to percolation inspired our study.

Y.Y.T. and A.V.E. acknowledge the funding from the Foundation for the Advancement of Theoretical Physics and Mathematics ``BASIS'', grant~20-1-1-8-1.

\appendix
\section{Mathematical background of the algorithm}\label{sec:BG}

Let $L_{1} ,L_{2} \in \mathbb{N}^\ast$, $p\in [ 0,1]$, $q=1-p$, $N=L_{1} L_{2}$. Define an experiment. Let $S$ be a table $L_{1} \times L_{2}$ that  consists of cells. The values of these cells are independent and identically distributed (i.i.d.) random variables $ S_{i,j} ( 1\leqslant i\leqslant L_{1} ,1\leqslant j\leqslant L_{2})$ obeying a Bernoulli distribution with parameter $p$. The experiment results in the table $S$, which cells got their values.

This experiment defines a probability space (a probability triple) $(\Omega,\mathcal{B}, \mathbb{P})$, where the sample space $\Omega $ is a set of tables $L_{1} \times L_{2}$ whose cells are $\{0,1\}$. We name the  possible outcomes as states. The $\sigma$-algebra $\mathcal{B}$ is a set of all the subsets $B\subset \Omega $. We name the elements of the $\sigma$-algebra as \emph{events}.
Let $S\in \Omega $, then assume $\mathbb{P}( S) =p^{\#_{1}( S)} q^{\#_{0}( S)}$, where $\#_{1}( S)$ and $\#_{0}( S)$ are the numbers of 1 and 0 values, respectively, in table~$S$.
\begin{remark}
We have defined the elementary outcomes as a set of cells in the table. However, in fact, all the reasoning in this work depends only on the number of cells, and does not depend on how exactly these cells are arranged. We need a shape for the table only for the visual convenience of applying the results of this work to the site percolation.
\end{remark}

Let $A$ be a non-decreasing event. We can describe in general terms a whole family of algorithms (or one algorithm, but with many parameters) for calculating the explicit formulas $\mathbb{P}(A)$ depending on the parameter~$p$.

\begin{proposition}
Probability $\mathbb{P}( A)$ of the event $A\in \Omega $ is a polynomial $R_{N,A}( p)$ depending on parameter $p$. The degree of the polynomial is $N$. This polynomial can be uniquely represented as a homogeneous polynomial of two variables $R_{N,A}( p,q) =\sum _{k=0}^{N} c_{k} p^{k} q^{N-k}$.
The coefficients of the polynomial $R_{N,A}( p,q)$ have clear combinatorial sense, viz.,  $c_{k}$ is the number of states, those that obey the event $A$ with exactly $k$ cells occupied by values of~1. The coefficients of the polynomial  $R_{N,A}( p,q)$ and $R_{N,A}( p)$ are connected via a linear relationship
$$
c_{n} =\sum _{k=0}^{n} b_{k}
\left(
  \begin{array}{c}
    N-k \\
    N-n \\
  \end{array}
\right)
.
$$
\end{proposition}

\begin{definition}
Let $B$ be a table $L_{1} \times L_{2}$ whose cells are $\{0,1,?\}$. The table $B$ uniquely identifies the event $\{S\in \Omega \mid \forall ( i,j) ( b_{i,j} \neq ?) \Rightarrow ( s_{i,j} =b_{i,j})\}$. We define such an event as a \emph{situation}, which we will identify with the table $B$ that generates it. Thus, a set of situations is identified with a set of tables of size $L_{1} \times L_{2}$ whose cells are $\{0,1,?\}$. The support of the situation is the set of cells not occupied by the symbol~$?$.
\end{definition}

\begin{definition}
Situations $S_{1}$ and  $S_{2}$ are  equivalent with respect to the event $A$, if  $\mathbb{P}( A\mid S_{1}) =\mathbb{P}( A\mid S_{2})$. The corresponding notation is $S_{1} \sim S_{2} ( A)$ or, more simply, $S_{1} \sim S_{2}$,  when it is clear from the context what kind of event $A$ is under consideration.
\end{definition}

Let us number all the cells of the table $L_{1} \times L_{2}$ as $z_{1} ,z_{2} ,\dots ,z_{N}$.
\begin{definition}
  $B_{0}$ is a set of events-situations with support $\varnothing $, $B_{0} =\left\{??\right\}$, where $??$ denotes the situation when each cell is occupied by the sign $?$ (question mark).
$B_{k}$ is a set of events-situations with support  $\{z_{1} ,z_{2} ,\dots ,z_{k}\}$ ($1\leqslant k\leqslant N$).
\end{definition}

Let $A$ be an event, then $\mathbb{P}( A) =\mathbb{P}\left( A\mid ??\right)$. Thus, to find $\mathbb{P}( A)$, the conditional probability should be found $\mathbb{P}\left( A\mid ??\right)$.

\begin{definition}
  A layer of order $k$ with respect to the event $A$ is a triplet $\mathcal{L}_{k}( A) =( L_{k} ,\{b_{B}( p,q) \mid B\in L_{k}\} ,g_{k}( p,q))$, where $L_{k} \subset B_{k}$, while $ b_{B}( p,q)$, ($B\in L_{k}$) and $g_{k}( p,q)$ are the homogeneous polynomials of degree $k$, such that $ \mathbb{P}( A) =\sum _{S\in L_{k}}\mathbb{P}( A\mid S) b_{S}( p,q) +g_{k}( p,q)$.
\end{definition}

\begin{remark}
This definition does not describe a layer uniquely. Anything that meets this definition may be denoted as a layer.
\end{remark}

\begin{proposition}
When a layer $\mathcal{L}_{k}( A)$ is given, then computation of $\mathbb{P}( A)$ reduces to  computation of the conditional probabilities $\mathbb{P}( A\mid S)$, ($S\in L_{k}$).
\end{proposition}

Let $L_{0} =B_{0}$, $b_{??}( p,q) \equiv 1$ and $g_{0}( p,q) \equiv 0$, i.e., the layer $\mathcal{L}_{0}$ is given.
\begin{proposition}
When a layer $\mathcal{L}_{0}$ is given, computation of $\mathbb{P}( A)$ reduces to a computation of the conditional probabilities $\mathbb{P}( A\mid S)$, ($S\in L_{0}$):
$\mathbb{P}( A) =\sum _{S\in L_{0}}\mathbb{P}( A\mid S)  b_{S}( p,q) +g_{0}( p,q)$.
\end{proposition}
\begin{definition}
Let $B$ be a situation. $B|_{i,j\rightarrow x}$ is the situation, that is given by the table $B$, in which the cell $b_{i,j}$ is replaced by $x$.
\end{definition}
\begin{definition}
Let $k:0\leqslant k< N$ and $S\in B_{k}$, then $S_{+} =S|_{z_{k+1}\rightarrow 1}$ and $S_{-} =S|_{z_{k+1}\rightarrow 0}$.
\end{definition}
\begin{proposition}
$S_{+} ,S_{-} \in B_{k+1}$.
\end{proposition}
\begin{proposition}
Let $k:0\leqslant k< N$ and $S\in B_{k}$, then
$\mathbb{P}( A\mid S) =\mathbb{P}( A\mid S_{+}) p+\mathbb{P}( A\mid S_{-}) q$.
\end{proposition}
\begin{definition}
$L_{k}^{+} =\{S_{+} \mid S\in L_{k}\}$,
$L_{k}^{-} =\{S_{-} \mid S\in L_{k}\}$.
\end{definition}
\begin{definition}
Any situation $S$ is unambiguously comparable to its status in relation to the event$A$:
$\mathbb{P}( A\mid S) =1$ (situation $S$ has a positive status),
$\mathbb{P}( A\mid S) =0$ (situation $S$ has a negative status),
$\mathbb{P}( A\mid S) \notin \{0,1\}$ (situation $S$ has a neutral status).
\end{definition}

Let we have Algorithm~1, that, for any situation $S$, determines its status with respect to the event $A$.
\begin{remark}\label{rem:3}
When $S \in B_{N}$, Algorithm~1 gives positive, negative or neutral status with respect to the event $A$.
\end{remark}

Let we have Algorithm~2, which $\forall S_{1} ,S_{2} \in B_{k}$ makes a quick check for their equivalence and either confirms their equivalence ($S_{1} \sim S_{2}$) or provides no new information ($S_{1} \sim S_{2}$ or $S_{1} \nsim S_{2}$).
\begin{remark}
There can be many such algorithms (e.g., a trivial algorithm that never produces any new information). We will need Algorithm~1 to determine the equivalence of the situations, and can also do it in a sense `fast enough' (if the equivalence check takes too long, then we prefer to abandon it).
\end{remark}

\begin{proposition}
Using Algorithm~1 and Algorithm~2 represent the set  $L_{k}^{+} \sqcup L_{k}^{-}$ as
$$
L_{k}^{+} \sqcup L_{k}^{-} =\left( \bigsqcup_{B\in L_{k+1}} H_{B}\right) \sqcup M_{k}^{0} \sqcup M_{k}^{1},
$$
where the sets $H_{B}$ are such that $\forall S\in H_{B} :\mathbb{P}( A\mid S) \notin \{0,1\}$;
$B\in H_{B}$;
$\forall S \in H_{B} :S\sim B$,
$M_{k}^{1} =\left\{S\in L_{k}^{+} \cup L_{k}^{-} \mid \mathbb{P}( A\mid S) =1\right\}$,
$M_{k}^{0} =\left\{S\in L_{k}^{+} \cup L_{k}^{-} \mid \mathbb{P}( A\mid S) =0\right\}$.
\end{proposition}

\begin{algorithmic}[1]
\State $ M_{k}^{1} \gets \varnothing $
\State $ M_{k}^{0} \gets \varnothing $
\State $ L_{k+1} \gets \varnothing $
\For{ $ S\in L_{k}^{+} \sqcup L_{k}^{-}$}
\State	$ status \gets \Call{Algoritm1}{S}$
	\If{ $ status=1$}
\State 		add $ S$ to $ M_{k}^{1}$
	\ElsIf {$ status=-1$}
\State 		add $ S$ to $ M_{k}^{0}$
	\Else
\State 		$ flag \gets 0$
		\For{ $ B\in L_{k+1}$}
 			\If{ $ \Call{Algoritm2}{S,B} =1$ }
\State 				add $ S$ to $ H_{B}$
\State 				$ flag \gets 1$
\State 					break for
\EndIf
\EndFor
		\If{ $ flag=0$ }
\State 			$ H_{S} \gets \{S\}$
\State 			add $ S$ to $ L_{k+1}$
\EndIf
\EndIf
\EndFor

\end{algorithmic}

\begin{definition}
  $$
h_{S}( p,q) =
\left\{
\begin{array}{lcl}
  b_{S}( p,q) p, &  & S \in L_{k}^{+}; \\
  b_{S}( p,q) q, & & S \in L_{k}^{-}.
\end{array}
\right.
$$
$b_{B}( p,q) =\sum _{S\in H_{B}} h_{S}( p,q)$.
$g_{k+1}( p,q) =$$\sum _{S\in M_{k}^{1} } h_{S}( p,q) +g_{k}( p,q) ( p+q)$.
\end{definition}
Let $k:0\leqslant k< N$ and a layer is given $\mathcal{L}_{k}( A)$: $\mathbb{P}( A) =\sum _{S\in L_{k}}\mathbb{P}( A\mid S) b_{S}( p,q) +g_{k}( p,q)$.
Build up a layer $\mathcal{L}_{k+1}$ and thus simplify the problem of calculating $\mathbb{P}( A)$:
\begin{eqnarray*}
\fl \mathbb{P}( A) =\sum _{S\ \in \ L_{k}}\mathbb{P}( A\mid S)  b_{S}( p,q) +g_{k}( p,q) =\\
=\sum _{S\ \in \ L_{k}}(\mathbb{P}( A\mid S_{+})  p+\mathbb{P}( A\mid S_{-})  q)  b_{S}( p,q) +g_{k}( p,q)  ( p+q) =\\
=\sum _{S\ \in \ L_{k}}\mathbb{P}( A\mid S_{+})  p b_{S}( p,q) +\sum _{S\ \in \ L_{k}}\mathbb{P}( A\mid S_{-})  q b_{S}( p,q) +g_{k}( p,q)  ( p+q) =\\
=\sum _{S\ \in \ L_{k}^{+}}\mathbb{P}( A\mid S)  p b_{S}( p,q) +\sum _{S\ \in \ L_{k}^{-}}\mathbb{P}( A\mid S)  q b_{S}( p,q) +g_{k}( p,q)  ( p+q) =\\
=\sum _{S\ \in \ L_{k}^{+} \sqcup L_{k}^{-}}\mathbb{P}( A\mid S)  h_{S}( p,q) +g_{k}( p,q)  ( p+q) =\\
=\sum _{B\ \in \ L_{k+1}}\mathbb{P}( A\mid B)  \sum _{S\ \in \ H_{B}} h_{S}( p,q) + \sum _{S\ \in M_{k}^{1} \ } h_{S}( p,q) +g_{k}( p,q)  ( p+q)  =\\
=\sum _{B\in \ L_{k+1}}\mathbb{P}( A\mid B)  b_{B}( p,q) +g_{k+1}( p,q) .
\end{eqnarray*}
As the result, we get a layer $\mathcal{L}_{N}$: $\mathbb{P}( A) =\sum _{S\in L_{N}}\mathbb{P}( A\mid S) b_{S}( p,q) +g_{N}( p,q)$ (see Remark~\ref{rem:3}). In this way, we can compute  $\mathbb{P}( A)$.

In order to use this method in practice to calculate the explicit formula $\mathbb{P}( A)$ with respect to the parameter $p$, one needs to (i)~set $L_{1}, L_{2} \in \mathbb{N}^\ast$; (ii)~set a non-decreasing event $A$; (iii)~number all table cells $L_{1} \times L_{2}$ such that $z_{1}, z_{2}, \dots , z_{N}$; (iv)~put in Algorithm~1; (v)~put in Algorithm~2.

\section*{References}
\bibliographystyle{iopart-num.bst}

\bibliography{polynomials,scaling,newbib}

\end{document}